\documentclass[submission,copyright,creativecommons]{eptcs}
\providecommand{\event}{NCL'26} 

\usepackage{iftex}

\ifpdf
  \usepackage{underscore}         
  \usepackage[T1]{fontenc}        
\else
  \usepackage{breakurl}           
\fi

\usepackage{philex}
\usepackage{amssymb}
\usepackage{latexsym}
\usepackage{amsthm}
\usepackage{bussproofs}
\usepackage{url}
\usepackage{forest}
\usepackage{stackengine}
\stackMath

\usepackage{rotating}\newcommand{\invertediota}{\begin{sideways}\begin{sideways}$\iota$\end{sideways}\end{sideways}}

\newtheorem{theorem}{Theorem}

\usepackage{tikz}
\usepackage{tikz-qtree}

\def\titlerunning{Tableaux for Theories of Definite Descriptions}
\def\authorrunning{Nils Kürbis}
\begin{document}
\title{Tableaux Rules for Theories of Definite Descriptions in Positive Free Logic with Some Philosophical Remarks}
\author{Nils K\"urbis
\institute{Department of Logic, University of Lodz}
\email{nils.kurbis@filhist.uni.lodz.pl}
}
\date{}
\maketitle


\begin{abstract} 
\noindent While Russell's theory of definite descriptions, expressions of the form `the $F$', is almost incontrovertible in classical logic, in free logic matters stand differently. Several alternatives have been proposed. In Lambert's hierarchy, the weakest theory is his minimal free theory of definite descriptions ($\mathbf{MFD}$), which only says that anything is the $F$ iff it is the unique existing $F$, the strongest one is a chosen object theory in which all improper definite descriptions refer to the same object ($\mathbf{FD2}$). Intermediate is Lambert's $\mathbf{FD1}$, which adds $\invertediota x(x=t)=t$ as an axiom to $\mathbf{MFD}$. $\mathbf{MFD}$ is very weak. $\mathbf{FD2}$ is very strong. One might hope to find intermediate systems in which $\invertediota x(x=t)=t$ is provable from more fundamental principles. The present paper considers one avenue where one might hope to find them and closes it. Two theories, $\mathbf{FDA}$ and $\mathbf{FDH}$, will be considered that are motivated by considerations from proof-theoretic semantics. The former turns out to be equivalent to $\mathbf{FD2}$, the latter to have unreasonable consequences. A third theory, similarly motivated, is also not viable. The philosophical upshot of this is briefly discussed. 
\end{abstract}



\section{Introduction}
Definite descriptions, expressions of the form `the $F$', are ubiquitous in ordinary discourse and reasoning. Here are just three examples:\footnote{Adapted from \cite {russellondenoting} and \cite{russellknowacdesc}.}\medskip

\noindent George IV. enquired whether Scott was the author of Waverley, but he did not enquire whether Scott was Scott.\smallskip

\noindent The returning officer knows all the candidates and that the candidate with the most votes wins the election, but does not know who wins the election.\smallskip

\noindent The candidate with the most votes wins the election. The Unionist candidate wins the most votes. Sir John is the Unionist candidate. Therefore, Sir John wins the election.\medskip

\noindent Grammatically, definite descriptions are singular terms and in the same category as proper names: `The $F$' purports to refer to the unique object that answers to the description $F$, and it forms a sentence from a predicate `is $G$'. According to the most celebrated theory of definite descriptions, Russell's, however, while this may be so grammatically, from the logical and semantical perspective it is only apparent. According to Russell, `The $F$ is $G$' means `There is exactly one $F$ and it is $G$'. Thus the definite description disappears upon analysis and need neither be considered in its own right logically nor semantically. Russell concludes that definite descriptions have meaning only in the context of complete sentences, and that by themselves they are meaningless.\footnote{See \cite{russellondenoting} for the first statement of the theory, \cite[Introduction, Ch. 3, and Part I, Sec. B, $\ast 14$]{PM1} for its formal development, and \cite{russellmathematicalphilosophy} for further commentary. For a general overview of theories of definite descriptions, see \cite{kurbisdefdescr}.}

When he first presented his theory 120 years ago, Russell commented that his analysis `may seem a somewhat incredible interpretation' \cite[482]{russellondenoting}, and despite its status as a `paradigm of philosophy' \cite[1, footnote 1]{ramseyphilpapers}, Russell's conclusions have been criticised and resisted. In particular, it is felt that definite descriptions are singular terms with perfectly good meanings. 

A motivation for the development of free logic by Hintikka, Lambert and others was the formalisation of theories of definite descriptions that are alternatives to Russell's, and in which definite descriptions are genuine singular terms, grammatically, logically and semantically.\footnote{For the original papers, see \cite{hintikkaexential,hintikkatowardsdefdesc}, \cite{lambertnotesEII,lambertnotesEIII,lambertnotesEIV}, the latter conveniently summarised and updated in \cite{lambertEThe} and developed further in \cite{lambertfreedef}. For general overviews over free logic, see \cite{bencivengahandbook} and \cite{morschersimonsfreelogic}.} 

Classical logic, the logic of Frege and Russell, demands that the domain of quantification is not empty and that every singular term refers. What is often called \emph{universally free logic} permits the empty domain of quantification as well as non-referring terms. At least that was the initial position of Lambert's. In what is now the standard semantics for free logic there are two domains, an inner domain over which the quantifiers range, often said to contain the existing objects and which may be empty, and an outer domain, often said to contain the non-existing objects, required to be non-empty, and it is to objects in the latter to which `non-referring' singular terms refer.\footnote{See \cite{cocchiarellaobjects}, \cite{meyerlambert}. This way of doing free logic has proved fruitful in quantified modal logic, see \cite{garsonmodallogic}, \cite[Ch. 16]{hughescresswell}, \cite{mendelsohnfittingFOML}; the systems of Hughes and Cresswell and Fitting and Mendelsohn are not universally free in that it is required that the domain of every world contains at least one object.} I shall set this issue aside. My approach in the present paper will be largely proof-theoretical. I shall give tableaux rules for a number of theories of definite descriptions in universally free logic. 

Call a definite description `the $F$' \emph{proper} if there exists exactly one $F$ that answers to the description.\footnote{This definition is unapologetically Russellian, but it does not appear that anyone objects.} In this case `the $F$' refers to this object. The question is what to do with the improper ones. Lambert proposed three major theories which form a linear hierarchy in terms of strength. The weakest theory is the minimal free theory of definite descriptions ($\mathbf{MFD}$), which only says that anything is the $F$ iff it is the unique existing $F$, the strongest one is a chosen object theory in which all improper definite descriptions have the same referent ($\mathbf{FD2}$). Intermediate between them is Lambert's $\mathbf{FD1}$, which adds $\invertediota x(x=t)=t$ as an axiom to $\mathbf{MFD}$. 

$\mathbf{MFD}$ is very weak. $\mathbf{FD2}$ is very strong. Lambert regards $\invertediota x(x=t)=t$ as an important principle, and it has considerable intuitive appeal. One might hope to find intermediate systems, in particular ones in which $\invertediota x(x=t)=t$ is provable from more fundamental principles. The present paper considers one avenue where one might hope to find suitable alternatives and presents tableaux rules for Lambert's three major theories and two others that naturally arise. 

The two further theories, $\mathbf{FDA}$ and $\mathbf{FDH}$, are motivated by two factors: (1) to prove $\invertediota x(x=t)=t$ from more fundamental principles; (2) to redress a certain misbalance in the rules of natural deduction governing $\invertediota$ in $\mathbf{MFD}$. Both systems have a philosophical motivation in proof-theoretic semantics, as I will briefly indicate. $\mathbf{FDH}$ also at first glance looks like a minimal intervention that fixes the problem that the theory initially proposed by Hintikka was inconsistent by giving his axiom an ever so slightly different formulation. While Hintikka proposed that $t$ is the $F$ iff $t$ is $F$ and anything that is $F$ is $t$, $\mathbf{FDH}$ is based on the idea that $t$ is the $F$ iff $t$ is the unique $F$. $\mathbf{FDA}$ is an amalgamate of $\mathbf{MFD}$ and $\mathbf{FDH}$ in that it employs only one half of each of their characteristic axioms. $\mathbf{FDA}$ will turn out to be equivalent to $\mathbf{FD2}$, $\mathbf{FDH}$ too strong to be reasonable. This necessitates a correction recently made concerning these two systems \cite{kurbisiotaposfreelogic}. Yet another system is considered in the conclusion, also with a proof-theoretic motivation, which is not viable as a theory of definite descriptions either. All this has a philosophical upshot that is of some interest, which is discussed briefly in the conclusion. The avenue in which one might have hoped to find alternative theories of definite descriptions has been closed. 

Lambert and others proposed a number of further systems. Lambert reports that van Fraassen suggested that there was a linearly ordered hierarchy of theories of definite descriptions in free logic, but this proved to be wrong \cite[69, footnote $\ast$; 90]{lambertfoundations}. The situation presented by Lambert requires two dimensions of comparison \cite[90f]{lambertfoundations}: it is a spectrum rather than a hierarchy, as van Fraassen put it. The spectrum deserves more attention and a motivation of the present paper is to draw attention to it.

\section{Tableaux Rules for Positive Free Logic}
The language is a standard language of first order logic with identity, except that it contains the predicate $\exists !$ and the $\invertediota$ operator for definite descriptions. $\exists !t$ expresses `$t$ exists' or `$``t''$ refers'. $\invertediota$ forms a singular term $\invertediota xF$ from a variable $x$ and formula $F$, binding $x$. Although this is not strictly speaking necessary, I shall employ the convention, common in proof theory, to distinguish free and bound variables graphically, $x, y, z$ for the latter, $a$, $b$, $c$ for the former. Free variables are also called \emph{parameters}. The singular terms, or terms for short, are the constants, parameters and $\invertediota$ terms. I shall leave out functional terms for simplicity, but as they can be defined in terms of $\invertediota$, as do Whitehead and Russell, nothing much is lost. $t$, $t_1$, $t_2$ etc. are used for terms in general. 

The tableaux proof system is standard, $\vdash$ defined as usual. I follow the format used by Priest, who gives these rules for classical positive free logic \cite[Chs 1.4, 12.5, 13.3]{priestnonclassical}: 

\begin{center} 
\Tree
[.$A\rightarrow B$ [.$\neg A$ ] [.$B$ ] ]
\qquad
\Tree
[.$\neg (A\rightarrow B)$ [.$\stackunder{A}{\neg B}$ ] ]
\qquad\qquad
\Tree 
[.$\neg\neg A$ [.$A$ ] ]
\qquad\qquad
\Tree
[.$A\leftrightarrow B$ [.$\stackunder{\neg A}{\neg B}$ ] [.$\stackunder{A}{B}$ ] ]
\qquad
\Tree
[.$A\leftrightarrow B$ [.$\stackunder{\neg A}{B}$ ] [.$\stackunder{A}{\neg B}$ ] ]

\bigskip

\Tree
[.$\forall xA$ [.$\neg \exists !t$ ] [.$A_t^x$ ] ]
\qquad
\Tree
[.$\neg\forall xA$ [.$\stackunder{\exists !a}{\neg A_a^x}$ ] ]
\qquad\qquad
\Tree
[. $\stackunder{A_{t_1}^x}{t_1=t_2}$ [.$\stackunder{\vert}{A_{t_2}^x}$ ] ]
\qquad
\Tree
[.  $\centerdot$ [.$t=t$ ] ]
\end{center} 

\noindent where $t$ is any term and $a$ is new to the branch.\medskip

\noindent The rules for the biconditional may be derived, but treating them as primitive is useful. 

Garson gives slightly different rules for the quantifiers \cite[Ch 14.1]{garsonmodallogic}, typographically adjusted: 

\begin{center}
\begin{forest}
[$\stackunder{\forall xA}{\exists !c}$
[$A_c^x$
]]
\end{forest}\qquad
\begin{forest}
[$\neg \forall xA$
[$\stackunder{\exists !a}{\neg A_a^x}$
]]
\end{forest}
\end{center}

\noindent where $c$ is a constant and $a$ is new to the branch.\medskip
\newpage

\noindent Garson does not have parameters in his language, but clearly we can use them instead of constants, too. Garson restricts the positive $\forall$ rule due to complications with definite descriptions in modal logic: they may be non-rigid, while constants are rigid. As we are not doing modal logic, we can lift the restriction and permit any term replacing $c$. Garson's rule has a better branching factor than Priest's, and I shall use it instead.  The resulting system is equivalent to Lambert's system of positive free logic. Call it $\mathbf{CPF}$. 

A neater result would be achieved if instead we use a method introduced by Indrzejczak in the context of cut free sequent calculi for free logics \cite{andrzejcutfreefreelogic}. Permit only parameters and constants in the positive $\forall$ rule, and add the rule $(\exists !E)$: 

\begin{center}
\Tree 
[.$\exists !\invertediota xF$ [.$a=\invertediota xF$ ]]
\end{center}

\noindent where $a$ is new to the branch.\medskip

\noindent This rule is also used in Indrzejczak's and Zawidzki's tableaux systems for free logics \cite{indrzejczakzawidzkitableauDD}. With Priest's rules, there is no hope for the subformula property, with $(\exists !E)$ there is at least a version of it if $\exists !$ and $=$ are discounted. Priest's positive $\forall$ rule is then derivable, as is the version of Garson's rule which permits $\invertediota$ terms for the instantiation.\footnote{I'm indebted to Andrzej Indrzejczak for the suggestion in this paragraph. Indrzejczak observed that the sequent calculus version of $(\exists !E)$ ensures that a problem that occurs with the replacement of parameters by $\invertediota$ terms in proofs of cut elimination (admissibility) does not arise: in a nutshell, when parameters are replaced by $\invertediota$ terms, they may introduce cut formulas of arbitrary degree when a quantified cut formula is removed. A similar problem arises in proofs of normalisation, and Indrzejczak's trick solves the problem here, too \cite{kurbisiotanegfreelogic,kurbisiotaposfreelogic}.}

\section{Lambert's Theories of Definite Descriptions}
In his earliest papers on the subject Lambert proposed three free theories of definite descriptions of increasing strength. The following sections give tableaux rules for them. 

\subsection{The Minimal Free Theory of Definite Descriptions}
Lambert's \emph{minimal free description} theory $\mathbf{MFD}$ is formalised by adding what has come to be called \emph{Lambert's Law} as an axiom to $\mathbf{CPF}$: 

\lbp{LA}{$LA$}{$\forall z(\invertediota xA=z\leftrightarrow \forall y(A_y^x\leftrightarrow y=z))$}

\noindent $\mathbf{MFD}$ has proved to be quite popular. It is the theory that underlies Garson's treatment of definite descriptions \cite[Ch. 18]{garsonmodallogic}. Garson does not formulate separate tableaux rules for $\invertediota$, but simply permits the addition of \rf{LA} to a tableau at any point \cite[400]{garsonmodallogic}. Priest does not consider the formalisation of definite descriptions in \cite{priestnonclassical}. It is not difficult to devise tableaux rules for $\invertediota$ that do not merely permit to add \rf{LA}:\footnote{The rules given here differ slightly from those given by Indrzejcak and Zawidzki \cite{indrzejczakzawidzkitableauDD}, who aim for a unified treatment of $\mathbf{MFD}$ and its variations in other kinds of free logic. The theories of definite descriptions in \cite[Ch. 17]{hughescresswell} and \cite[Ch. 20]{mendelsohnfittingFOML} require scope markers, which the theories of Lambert et al. tried to avoid. They are essentially Russell's theory transposed to modal logic. The former are explicit about this and follow Smullyan's treatment of definite descriptions in modal logic \cite{smullyanmoddesc}, who used Russell's notation for marking scope. The latter use the $\lambda$ operator for predicate abstraction; see \cite{recentdevelopmentsDD} for why they nonetheless remain Russellian. Tableaux rules for modal logic with an $\invertediota$ operator that interacts with a $\lambda$ operator  are also given in \cite{indrejcakzawidzkiDDhybrid}.}

\begin{center}
\begin{forest}
[$\stackunder{$\stackunder{\exists !t_1}{\exists !t_2}$}{\invertediota x A=t_1}$
[$\stackunder{A_{t_2}^x}{t_2=t_1}$]
[$\stackunder{\neg A_{t_2}^x}{\neg t_2=t_1}$]
]
\end{forest}\qquad\qquad
\begin{forest}
[$\stackunder{\exists !t_1}{\neg \ \invertediota xA=t_1}$
[$\exists !b$
[$\stackunder{A_b^x}{\neg b=t_1}$]
[$\stackunder{\neg A_b^x}{b=t_1}$]
]]
\end{forest}
\end{center}

\noindent where $b$ is new to the branch.\medskip 

\noindent These rules are remarkably simple, given that \rf{LA}, with its two biconditionals and two quantifiers, is quite complex. 

Call the systems that result from $\mathbf{CPF}$ by adding \rf{LA} as an axiom $\mathbf{MFD}^A$ and by adding the tableaux rules for $\invertediota$ above $\mathbf{MFD}^T$. 

\begin{theorem}\label{MFDEquiv}
$\Gamma\vdash A$ in $\mathbf{MFD}^A$ iff $\Gamma\vdash A$ in $\mathbf{MFD}^T$. 
\end{theorem}

\begin{proof}
To illustrate the workings of the system, we prove this theorem explicitly.  

(a) \rf{LA} is provable in $\mathbf{MFD}^T$ using the tableaux rules for $\invertediota$:

\begin{center}
\begin{forest}
[${\neg \forall z(\invertediota xA=z\leftrightarrow \forall y(A_y^x\leftrightarrow y=z))}$
[$\stackunder{\exists !a}{\neg (\invertediota xA=a\leftrightarrow \forall y(A_y^x\leftrightarrow y=a))}$
[$\stackunder{\invertediota xA=a}{\neg\forall y(A_y^x\leftrightarrow y=a)}$
[$\stackunder{\exists !b}{\neg(A_b^x\leftrightarrow b=a)}$
[$\stackunder{A_b^x}{\neg \ b=a}$
[$\stackunder{A_b^x}{b=a}$]
[$\stackunder{\neg A_b^x}{\neg \ b=a}$]]
[$\stackunder{\neg A_b^x}{b=a}$
[$\stackunder{A_b^x}{b=a}$]
[$\stackunder{\neg A_b^x}{\neg \ b=a}$]]
]]
[$\stackunder{\neg \ \invertediota xA=a}{\forall y(A_y^x\leftrightarrow y=a)}$
[${A_y^a\leftrightarrow a=a}$
[$\stackunder{A_a^x}{a=a}$
[$\exists !b$
[$\stackunder{A_b^x}{\neg \ b=a}$
[${A_b^x\leftrightarrow b=a}$
[$\stackunder{A_b^x}{b=a}$]
[$\stackunder{\neg A_b^x}{\neg \ b=a}$]]]
[$\stackunder{\neg A_b^x}{b=a}$
[${A_b^x\leftrightarrow b=a}$
[$\stackunder{A_b^x}{b=a}$]
[$\stackunder{\neg A_b^x}{\neg \ b=a}$]]]]]
[$\stackunder{\neg A_a^x}{\neg \ a=a}$
[{$a=a$}]]]]]]]]]]
\end{forest}
\end{center}

\noindent The four leftmost leaves are added in accordance with the positive $\invertediota$ rule and close. The rightmost branch closes because of the law of identity. In the branches between these, $\exists !b$ and the branches that come off it are added by the negative $\invertediota$ rule, and they close by applying rules for $\forall$ and $\leftrightarrow$. 

(b) The tableaux rules for $\invertediota$ are derivable in $\mathbf{MFD}^A$. 

\begin{center}
\begin{forest}
[$\stackunder{$\stackunder{\exists !t_1}{\exists !t_2}$}{\invertediota x A=t_1}$
[${\forall z(\invertediota xA=z\leftrightarrow \forall y(A_y^x\leftrightarrow y=z))}$
[${\invertediota xA=t_1\leftrightarrow \forall y(A_y^x\leftrightarrow y=t_1)}$
[$\stackunder{\invertediota xA=t_1}{\forall y(A_y^x\leftrightarrow y=t_1)}$
[${A_{t_2}^x\leftrightarrow t_2=t_1}$
[$\stackunder{A_{t_2}^x}{t_2=t_1}$]
[$\stackunder{\neg A_{t_2}^x}{\neg t_2=t_1}$]
]]
[$\stackunder{\neg \ \invertediota xA=t_1}{\neg\forall y(A_y^x\leftrightarrow y=t_1)}$]
]]]
]]]
\end{forest}\qquad
 \begin{forest}
[$\stackunder{\exists !t_1}{\neg \ \invertediota xA=t_1}$
[${\forall z(\invertediota xA=z\leftrightarrow \forall y(A_y^x\leftrightarrow y=z))}$
[${\invertediota xA=t_1\leftrightarrow \forall y(A_y^x\leftrightarrow y=t_1)}$
[$\stackunder{\invertediota xA=t_1}{\forall y(A_y^x\leftrightarrow y=t_1)}$]
[$\stackunder{\neg \ \invertediota xA=t_1}{\neg \forall y(A_y^x\leftrightarrow y=t_1)}$
[$\stackunder{\exists !b}{\neg (A_b^x\leftrightarrow b=t_1)}$
[$\stackunder{A_b^x}{\neg \ b=t_1}$]
[$\stackunder{\neg A_b^x}{b=t_1}$]
]]]]]]
\end{forest}
\end{center}

\noindent On the right, the rightmost branch closes and the other two constitute the derivation of the positive $\invertediota$ rule. On the left, the leftmost branch closes, and the other two constitute the derivation of the negative $\invertediota$ rule.\end{proof}

\subsection{FD1}
According to Lambert, $\mathbf{MFD}$ is too minimal as `an important principle evidently not deducible from [\rf{LA}] is the principle of \emph{Cancellation}' \cite[30]{lambertEThe}: 

\lbp{C}{$CA$}{$\invertediota x (x=t) = t$}

\noindent Lambert's system $\mathbf{FD1}$ results from adding \rf{C} to $\mathbf{MFD}$ (\cite[44]{lambertfreedef}, \cite[90f]{lambertfoundations}). Unfortunately there appears to be no other option but to add \rf{C} as an axiom to $\mathbf{MFD}$ also in its tableaux formulation, or alternatively to lay down that any branch containing $\neg \ \invertediota x (x=t) = t$ closes, neither of which are particularly exciting. 

It would be preferable to derive \rf{C} from some more fundamental principles captured in tableaux rules. This is an important motivation for investigating alternative systems. We'll encounter three failed attempts in the following. 

Call the systems that result from $\mathbf{CPF}$ by adding \rf{LA} and \rf{C} as axioms $\mathbf{FD1}^A$ and by adding the corresponding tableaux rules for $\invertediota$ $\mathbf{FD1}^T$. 

\begin{theorem}
$\Gamma\vdash A$ in $\mathbf{FD1}^A$ iff $\Gamma\vdash A$ in $\mathbf{FD1}^T$. 
\end{theorem}

\begin{proof}Trivial.\end{proof}

\subsection{FD2}
The strongest of Lambert's systems, $\mathbf{FD2}$, is formalised by adding the following axiom to $\mathbf{CPF}$ (\cite[87]{lambertnotesEIV}, \cite[237]{fraassenlambert}):\footnote{\label{FD2equiv}Lambert uses the version $\invertediota xA=t\leftrightarrow \forall x (x=t\leftrightarrow (A\land \forall y(A_y^x\rightarrow y=x)))$. They are equivalent because in positive free logic $\exists !a\vdash (A_a^x\land \forall y(A_y^x\rightarrow y=a))\leftrightarrow\forall y(A_y^x\leftrightarrow y=a)$.}

\lbp{FD2}{$FD2$}{$\invertediota xA=t\leftrightarrow \forall x (x=t\leftrightarrow \forall y(A_y^x\leftrightarrow y=x))$}

\noindent \rf{FD2} implies both \rf{LA} and \rf{C}, but is not implied by them. A consequence of \rf{FD2} is 

\lbp{CO}{$CO$}{$(\neg \exists ! t_1\land \neg \exists !t_2)\rightarrow t_1=t_2$}

\noindent Speaking semantically, in \rf{FD2} there is one object in the outer domain to which all terms refer that do not refer in the inner domain. This is Carnap's \emph{chosen object theory} of improper definite descriptions, a version of which was also endorsed by Frege in some of his writings.\footnote{See \cite[33, 41]{fregesinnbedeutung} (translated in \cite{geachblackfrege}), \cite[\S 7, \S 8]{carnapmeaningnecessity}, and also \cite{pelletierlinsky}.}

$\mathbf{FD2}$ can be axiomatised equivalently by adding \rf{CO} to $\mathbf{MFD}$ \cite{fraassenlambert}. It is trivial to turn the axiom into a tableaux rule: 

\begin{center}
\begin{forest}
[$\stackunder{\neg \exists !t_1}{\neg \exists !t_2}$
[${t_1=t_2}$]]
\end{forest}
\end{center}

\noindent Call the system that results from $\mathbf{CPF}$ by adding \rf{LA} and \rf{CO} as axioms $\mathbf{FD2}^A$ and call $\mathbf{FD2}^T$ the result of adding the corresponding tableaux rules for $\invertediota$. 

\begin{theorem}
$\Gamma\vdash A$ in $\mathbf{FD2}^A$ iff $\Gamma\vdash A$ in $\mathbf{FD2}^T$. 
\end{theorem}

\begin{proof}
(a) \rf{CO} is provable in $\mathbf{FD2}^T$ using the tableaux rules for $\invertediota$. (b) The tableaux rules for $\invertediota$ of $\mathbf{FD2}^T$ are derivable in $\mathbf{FD2}^A$. 
\end{proof}

\subsection{Lambert's Hierarchy}
The three systems of Lambert's form a hierarchy with $\mathbf{MFD}$ as the weakest, $\mathbf{FD2}$ as the strongest, and $\mathbf{FD1}$ inbetween, where the inclusions are proper: 

\begin{theorem}\label{lambertinclusions}
$\mathbf{MFD}\subset\mathbf{FD1}\subset\mathbf{FD2}$. 
\end{theorem}

\begin{proof}
See \cite{fraassenlambert}.\end{proof}

\noindent Lambert also considers a number of other theories of definite descriptions in free logic, proposed by himself, van Fraassen and others (see \cite[90f]{lambertfoundations} for a survey of the options that are on the table). The result is not a linear hierarchy, but comparison requires two dimension. This is not the place to discuss them all. Instead, I shall propose two further systems that may be motivated by considerations in proof-theoretic semantics, a theory not considered by Lambert.

\section{Two Systems Motivated by Proof-Theoretic Semantics}
\subsection{Proof-Theoretic Semantics}
According to proof-theoretic semantics, the meaning of a logical expression may be defined by the rules of inference governing it if they satisfy certain requirements: they must be \emph{in harmony}.\footnote{See \cite[Chs 10-13]{dummettLBM} and \cite{schroederheistersemantics}. The idea ultimately stems from a fertile remark of Gentzen's \cite[\S5.13, 189]{gentzenuntersuchungen}. For a brief overview and the present author's take on proof-theoretic semantics, see \cite{kurbisPTSNM} and \cite[Ch. 2]{kurbisproofandfalsity}.} Proof-theoretic semantics is first and foremost a theory of meaning for certain sentence-forming expressions, the sentential connectives and quantifiers. But it has been applied to term-forming operators, too, most notably by Tennant \cite{tennantabstraction}. The present considerations contributes to this project, although it must be said that they are not encouraging.\footnote{For more thoughts on these matters, see \cite{kurbisiotanegfreelogic,kurbisiotaposfreelogic}. Some of the claims in the latter are, however, confused or, in the light of the results to be reported, false, especially those in section 8.3 and 8.4. There is no space for further details.}

The notion of harmony is suggestive but difficult to spell out. For present purposes the following suffices: harmony consists in a balance between introduction and elimination rules of systems of natural deduction, where the introduction rules specify the grounds for the derivation of a formula with the expression in dominant position and the elimination rules their consequences. These grounds and consequences must be balanced so that all and only those consequences can be drawn that are justified by the grounds. Formally harmony is tied to the provability of a normalisation theorem\footnote{Via Prawitz's inversion principle \cite[33f]{prawitznaturaldeduction}.}, which is a necessary, but not a sufficient, condition for harmony: normalisation shows that only those consequences that are justified by the grounds may be drawn, but not that all of them can be. Further details need not concern us here. What has been said suffices to establish that the rules for $\invertediota$ in $\mathbf{MFD}$ are not harmonious. 

In natural deduction, the $\invertediota$ operator of $\mathbf{MFD}$ is governed by the following rules:\footnote{They modify Tennant's rules \cite{tennantabstraction,tennantnatural} for negative free logic to suit $\mathbf{CPF}$ and are equivalent to those given by Garson \cite[Ch. 18]{garsonmodallogic}, but their form here is preferable for proof-theoretic investigation.}

\begin{prooftree} 
\AxiomC{$\exists ! t$}
\AxiomC{$[a=t]^i \ [\exists ! a]^k$}
\noLine
\UnaryInfC{$\Xi$}
\noLine
\UnaryInfC{$F_a^x$}
\AxiomC{$[F_a^x]^j \ [\exists !a]^k$}
\noLine
\UnaryInfC{$\Pi$}
\noLine
\UnaryInfC{$a=t$}
\LeftLabel{$(\invertediota I)$ \ }
\RightLabel{$_{i, j, k}$}
\TrinaryInfC{$\invertediota xF=t$}
\end{prooftree}

\noindent where $a$ is not free in $F$ nor in $t$ and does not occur in any undischarged assumptions in $\Xi$ and $\Pi$ except those in the assumption classes displayed.

\begin{center} 
\AxiomC{$\invertediota xF=t$}
\AxiomC{$u=t$}
\AxiomC{$\exists !u$}
\LeftLabel{$(\invertediota E_1)$ \ } 
\TrinaryInfC{$F_u^x$}
\DisplayProof
\qquad
\AxiomC{$\invertediota xF=t$}
\AxiomC{$F_u^x$}
\AxiomC{$\exists ! t$}
\AxiomC{$\exists !u$}
\LeftLabel{$(\invertediota E_2)$ \ }
\QuaternaryInfC{$u=t$}
\DisplayProof
\end{center} 

\noindent The elimination rules are not in harmony with $(\invertediota I)$, because the latter requires a premise $\exists !t$ for the deduction of $\invertediota xF=t$, which cannot be retrieved from it by an elimination rule. Contrast with the rules for $\invertediota$ in negative free logic, in which precisely that inference is licensed by a third elimination rule $\invertediota xF=t\vdash \exists !t$. This rule cannot be added to positive free logic: As $\vdash \invertediota xA=\invertediota xA$, $\vdash \exists !\invertediota xA$, and so by $(\invertediota E_1)$ $A(\invertediota xA)$, which leads to inconsistency: replace $A$ by a contraction. 

\subsection{FDA}
To redress the misbalance in the rules governing $\invertediota$ in $\mathbf{MFD}$ and to restore harmony the option suggests itself to weaken the introduction rule instead of strengthening the elimination rules and to drop the problematic premise: 

\begin{prooftree} 
\AxiomC{$[a=t]^i \ [\exists ! a]^k$}
\noLine
\UnaryInfC{$\Xi$}
\noLine
\UnaryInfC{$F_a^x$}
\AxiomC{$[F_a^x]^j \ [\exists !a]^k$}
\noLine
\UnaryInfC{$\Pi$}
\noLine
\UnaryInfC{$a=t$}
\LeftLabel{$(\invertediota I')$ \ }
\RightLabel{$_{i, j, k}$}
\BinaryInfC{$\invertediota xF=t$}
\end{prooftree}

\noindent where $a$ is not free in $F$ nor in $t$ and does not occur in any undischarged assumptions in $\Xi$ and $\Pi$ except those in the assumption classes displayed.\medskip

\noindent Added to $\mathbf{CPF}$ in natural deduction, $(\invertediota I')$, $(\invertediota E_1)$ and $(\invertediota E_2)$ formalise the theory of definite descriptions called $\mathbf{FDA}$ in \cite{kurbisiotaposfreelogic}. It can be axiomatised by one half of \rf{LA} and an axiom formulated for all terms rather than by quantification: 

\lbp{HARL}{$HA^{rl}$}{$\forall y(A_y^x\leftrightarrow y=t)\rightarrow\invertediota xA=t$}

\lbp{LALR}{$LA^{lr}$}{$\forall z (\invertediota xA=z\rightarrow \forall y(A_y^x\leftrightarrow z=y))$}

\noindent \rf{HARL} is one half of a hypothetical axiom to be considered later, which explains its label. 
\rf{C} follows from the instance with $x=t$ for $A$ and principles of free logic, so the desideratum of deriving it from more fundamental principles is fulfilled. 

Corresponding tableaux rules are: 

\begin{center}
\begin{forest}
[$\stackunder{$\stackunder{\exists !t_1}{\exists !t_2}$}{\invertediota x A=t_1}$
[$\stackunder{A_{t_2}^x}{t_1=t_2}$]
[$\stackunder{\neg A_{t_2}^x}{\neg t_1=t_2}$]
]
\end{forest}
\qquad\qquad
\begin{forest}
[${\neg \ \invertediota xA=t_1}$
[$\exists !b$
[$\stackunder{A_b^x}{\neg b=t_1}$]
[$\stackunder{\neg A_b^x}{b=t_1}$]
]]
\end{forest}
\end{center}

\noindent where $b$ is new to the branch.\medskip

\noindent The negative rule for $\invertediota$ is the negative rule of $\mathbf{MFD}^T$ with the premise $\exists !t_1$ dropped. 

Call the systems that result from $\mathbf{CPF}$ by adding \rf{HARL} and \rf{LALR} as axioms $\mathbf{FDA}^A$ and by adding the tableaux rules for $\invertediota$ above $\mathbf{FDA}^T$. 

\begin{theorem}
$\Gamma\vdash A$ in $\mathbf{FDA}^A$ iff $\Gamma\vdash A$ in $\mathbf{FDA}^T$. 
\end{theorem}

\begin{proof}
(a) \rf{HARL} and \rf{LALR} are provable in $\mathbf{FDA}^T$ using the tableaux rules for $\invertediota$ in this section. (b) The tableaux rules for $\invertediota$ are derivable in $\mathbf{FDA}^A$. The proofs are similar to those for $\mathbf{MFD}$.
\end{proof}

\noindent Where does $\mathbf{FDA}$ stand in Lambert's hierarchy? The answer may be a little unexpected.\footnote{At least it was so to some. The result is due to referees of the present piece and the proof uses some of their ideas.} $\mathbf{FDA}$ is equivalent to $\mathbf{FD2}$. We will prove this axiomatically. 

\begin{theorem}\label{FD2Collapse}
$\Gamma\vdash A$ in $\mathbf{FDA}^A$ iff $\Gamma\vdash A$ in $\mathbf{FD2}^A$.
\end{theorem}

\begin{proof}
Part I. \rf{LALR} and \rf{HARL} are provable from \rf{FD2}. (a) Instantiating $t$ with $\invertediota xA$ in \rf{FD2} gives \rf{LA} from $\invertediota xA=\invertediota xA$, and \rf{LALR} follows from \rf{LA}. 

(b) To prove \rf{HARL} from \rf{FD2}, we prove the right hand side of \rf{FD2} from the left hand side of \rf{HARL} and detach. By Leibniz' Law (1) $\forall y(A_y^x\leftrightarrow y=t), a=t\vdash \forall y(A_y^x\leftrightarrow y=a)$. By free quantification theory, $\forall y(A_y^x\leftrightarrow y=a), \exists !a\vdash A_a^x\leftrightarrow a=a$, hence (a) $\forall y(A_y^x\leftrightarrow y=a), \exists !a\vdash A_a^x$, and also (b) $\forall y(A_y^x\leftrightarrow y=t), \exists !a\vdash A_a^x\leftrightarrow a=t$. From (a) and (b), (2) $\forall y(A_y^x\leftrightarrow y=t), \forall y(A_y^x\leftrightarrow y=a), \exists !a\vdash a=t$. From (1) and (2) $\forall y(A_y^x\leftrightarrow y=t), \exists !a\vdash a=t\leftrightarrow \forall y(A_y^x\leftrightarrow y=a)$, hence by free quantification theory $\forall y(A_y^x\leftrightarrow y=t)\vdash \forall x(x=t\leftrightarrow \forall y(A_y^x\leftrightarrow y=x))$. Hence by \rf{FD2}, $\forall y(A_y^x\leftrightarrow y=t)\vdash\invertediota xA=t$. 

Part II. \rf{FD2} is provable from \rf{LALR} and \rf{HARL}. \rf{LA} is immediate from \rf{LALR} and \rf{HARL} by principles of free logic. It suffices to prove \rf{CO}, as $\mathbf{FD2}$ is equivalently axiomatised by \rf{LA} and \rf{CO}. So assume $\neg \exists !t_1\land \neg \exists !t_2$. Consider the instance $(\ast)$ $\forall y(\neg \ y=y \leftrightarrow y=t_1)\rightarrow \invertediota x (\neg \ x=x)=t_1$ of \rf{HARL}. By Leibniz' Law $\neg \exists !t_1, a=t_1\vdash \neg \exists !a$, hence $\exists !a, \neg \exists !t_1, a=t_1\vdash \bot$, and so  $\exists !a, \neg \exists !t_1, a=t_1\vdash \neg \ a=a$. Conversely, $\neg \ a=a\vdash \bot$, hence $\neg \ a=a\vdash a=t_1$. So by the logic of $\leftrightarrow$, $\exists !a, \neg \exists !t_1\vdash \neg \ a=a \leftrightarrow a=t_1$, so by free quantification theory $\neg \exists !t_1\vdash \forall y(\neg y=y\leftrightarrow y=t_1)$. Thus by $(\ast)$ $\invertediota x (\neg \ x=x)=t_1$. Similarly $\invertediota x (\neg \ x=x)=t_2$, for any $t_2$. Hence $t_1=t_2$.
\end{proof}

\noindent In \cite{kurbisiotaposfreelogic}, a normalisation theorem is proved for $\mathbf{FDA}$. It is also claimed that $\mathbf{FD2}$ is not amenable to proof-theoretic semantics, because the introduction rule for $\invertediota$ given there does not at all compare to paradigm cases of harmonious rules. The claim has now been cast into doubt, because $\mathbf{FD2}$ is equivalent to $\mathbf{FDA}$. There is, however, a problem lurking here on which more in the conclusion. 

On a positive note, the tableaux rules given in this section for $\mathbf{FDA}$ are slightly simpler than those given earlier for $\mathbf{FD2}$, so a minor advance in the subject has been made.\footnote{The natural deduction rules, too, are simple and straightforward, as opposed to those for $\mathbf{FD2}$ in \cite{kurbisiotaposfreelogic}.}

\subsection{FDH} 
A penchant for symmetry may lead one to consider a system that arises by modifying the other direction of \rf{LA} just as \rf{HARL} already modifies one of them. The result is a biconditional: 

\lbp{HA}{$HA$}{$\invertediota xA=t\leftrightarrow \forall y(A_y^x\leftrightarrow y=t)$}

\noindent \rf{C} follows from the instance with $x=t$ for $A$ and general free logic, so the desideratum of deriving it from more fundamental principles is fulfilled. The corresponding tableaux rules are: 

\begin{center}
\begin{forest}
[$\stackunder{\exists !t_2}{\invertediota x A=t_1}$
[$\stackunder{A_{t_2}^x}{t_1=t_2}$]
[$\stackunder{\neg A_{t_2}^x}{\neg t_1=t_2}$]
]
\end{forest}\qquad\qquad
\begin{forest}
[${\neg \ \invertediota xA=t_1}$
[$\exists !b$
[$\stackunder{A_b^x}{\neg b=t_1}$]
[$\stackunder{\neg A_b^x}{b=t_1}$]
]]
\end{forest}
\end{center}

\noindent where $b$ is new to the branch.\medskip

\noindent Note the difference to $\mathbf{MFD}$: the existence of $t_1$ is not required. And that's all the difference. 

In \cite{kurbisiotaposfreelogic}, \rf{HA} was named after Hintikka, because of its resemblance to an axiom put forward in his first publication on our subject \cite{hintikkatowardsdefdesc}:

\lbp{HO}{$HO$}{$\invertediota xA=t\leftrightarrow (A_t^x\land \forall y(A_y^x\rightarrow y=t))$}

\noindent Lambert pointed out that \rf{HO} leads to inconsistency, as it implies that $A(\invertediota xA)$ is unconditionally assertible \cite{lambertnotesEIII}: let $t$ be $\invertediota xA$, detach with $\invertediota xA=\invertediota xA$, and simplify; but when $A$ is contradiction, $A(\invertediota xA)$, too, is a contradiction. 

In classical logic, there are two formulas that customarily are used as the Russellian analysis of `The $F$ is $G$': $\exists x(Fx\land \forall y(Fy\rightarrow x=y)\land Gx)$ and $\exists x(\forall y(Fy\leftrightarrow x=y)\land Gx)$. In classical logic, $(A_t^x\land \forall y(A_y^x\rightarrow x=t))$ and $\forall y(A_y^x\leftrightarrow y=t)$ are equivalent. In free logic, they are not. The latter is weaker, as $t$ is in the scope of a quantifier: to deduce $A_t^x$ requires the premises $\exists !t$. \rf{HA} may be thought to be a minimal adjustment of \rf{HO} that ensures consistency. I now, however, prefer to call it the hypothetical axiom rather than settling Hintikka with it. The hypothetical theory of definite descriptions axiomatised by it is indeed consistent, as follows from the normalisation theorem proved for what is there called `Hintikka's System' in \cite{kurbisiotaposfreelogic}, but it has another consequence that means that it is not a viable theory of definite descriptions. 

Call the systems that result from $\mathbf{CPF}$ by adding \rf{HA} as an axiom $\mathbf{FDH}^A$ and by adding the tableaux rules for $\invertediota$ above $\mathbf{FDH}^T$. 

\begin{theorem}\label{FDHEquiv}
$\Gamma\vdash A$ in $\mathbf{FDH}^A$ iff $\Gamma\vdash A$ in $\mathbf{FDH}^T$. 
\end{theorem}

\begin{proof}
(a) \rf{HA} is provable in $\mathbf{FDH}^T$ using the tableaux rules for $\invertediota$. (b) The tableaux rules for $\invertediota$ are derivable in $\mathbf{FDH}^A$. The proofs are similar to those for $\mathbf{MFD}$. 
\end{proof}

\noindent $\mathbf{FDH}$ has a motivation of sorts from symmetry and the desire to prove \rf{C} from more fundamental principles. Arguably, it is also motivated by proof-theoretic semantics. $(\invertediota E_2)$ requires a premise, $\exists !t$, to which there corresponds no discharged assumption above the second premise of $(\invertediota I')$. This suggests dropping this premise from $(\invertediota E_2)$: 

\begin{prooftree}
\AxiomC{$\invertediota xF=t$}
\AxiomC{$F_x^u$}
\AxiomC{$\exists !u$}
\LeftLabel{$(\invertediota E_2') \ $}
\TrinaryInfC{$u=t$}
\end{prooftree}

\noindent Adding $(\invertediota I')$, $(\invertediota E_1)$ and $(\invertediota E_2')$ to $\mathbf{CPF}$ in natural deduction formalises $\mathbf{FDH}$. As a normalisation theorem has been proved for this system, it satisfies a necessary requirement of proof-theoretic semantics. $\mathbf{FDH}$, however, cannot count as a satis factory theory of definite descriptions, or anything else for that matter, because of the following:\footnote{This result is also due to referees for this paper.}

\begin{theorem}\label{onlyonething}
In $\mathbf{FDH}$ $\vdash \forall x \forall y \ x=y$.
\end{theorem}

\begin{proof}
Consider the instance of the left to right half of \rf{HA} in which $A$ is $x=x$ and $t$ is $\invertediota x(x=x)$: $\invertediota x(x=x)=\invertediota x(x=x)\rightarrow \forall y(y=y\leftrightarrow y=\invertediota x(x=x))$. As $\invertediota x(x=x)=\invertediota x(x=x)$, $\forall y(y=y\leftrightarrow y=\invertediota x (x=x))$. So for any parameter $a$, $\exists !a\vdash a=a\leftrightarrow a=\invertediota x(x=x)$, hence $\exists !a\vdash a=\invertediota x(x=x)$, and for any different $b$, $\exists !b\vdash b=\invertediota x(x=x)$. Hence $\exists !a, \exists !b\vdash a=b$, and so $\vdash \forall x\forall y \ x=y$ by free quantification theory. 
\end{proof}

\noindent $\mathbf{FDH}$ permits at most one thing to exist. It is not a reasonable theory of definite descriptions, nor of anything else one might think is formalised by $\invertediota$, and this despite of its proof-theoretic credentials. This requires some discussion. 

Appealing only the to left to right half of \rf{HA}, Theorem \ref{onlyonething} also excludes another theory one may be inclined to consider, axiomatised by the right to left half of \rf{LA} and the left to right half of \rf{HA}. Fortunately it lacks proof-theoretic motivation.

\section{Concluding Philosophical Remarks}
The proof of Part II of Theorem \ref{FD2Collapse} requires only resources available in $\mathbf{FDH}$, so $\mathbf{FDH}$ includes $\mathbf{FD2}$, and we have the following hierarchy: $\mathbf{MFD}\subset\mathbf{FD1}\subset\mathbf{FD2}=\mathbf{FDA}\subset\mathbf{FDH}$. The last inclusion is proper, because $\mathbf{FD2}$ permits more than one object to exist.

A normalisation theorem can be proved for $\mathbf{FDA}$ and $\mathbf{FDH}$ in their natural deduction formulations. Both theories therefore satisfy a necessary condition of proof-theoretic semantics. $\mathbf{FDA}$ aka $\mathbf{FD2}$ has its proponents. $\mathbf{FDH}$ is not a viable theory of anything, because it restricts the domain of quantification to at most one thing. This suggests that $(\invertediota E_2')$ of $\mathbf{FDH}$ is too strong relative to $(\invertediota I')$: it permits the deduction of more than is justified by the latter. Yet this strength does not appear to stem from any kind of misbalance. The rules for $\invertediota$ of $\mathbf{FDH}$ are as harmonious as any. The premises of its elimination rules correspond exactly to the discharged assumptions of its introduction rule. Grounds and consequences should be balanced. $\mathbf{FDA}$, on the other hand, requires an additional premise in $(\invertediota E_2)$: it requires \emph{too much} for its application, and so fewer conclusions than are justified by the introduction rule can be drawn from it. A misbalance, then, lies in the rules for $\mathbf{FDA}$. Thus it is $\mathbf{FDH}$, not $\mathbf{FDA}$, that appears to satisfy also the sufficient conditions. Proof-theoretic semantics usually demands that sameness of introduction rules means sameness of elimination rules. This is not given here. But the problem is that it is the rules of $\mathbf{FDH}$ that are harmonious, and the only option would be to discard $\mathbf{FD2}$ on these grounds. 

$\mathbf{FDH}$ is not conservative over $\mathbf{CPF}$. I favour a view according to which conservativeness is a consequence of harmony (cf. Dummett's conjecture \cite[290]{dummettLBM}).\footnote{There may be exceptions. Prawitz \cite{prawitzreviewLBM} gives the truth predicate added to arithmetic as a counterexample, although this may have to do with the numbers. Read does not think that harmony implies conservativeness \cite{readharmony}. See also \cite{weisskurbislogic}.} In a nutshell, the reason is that if harmony obtains, we only take out of a sentence by an elimination rule what has been put into it according to its introduction rules, and so it is plausible that harmony implies conservativeness, as the latter would appear to require to take out more than what has been put in. The situation is thus rather awkward. Once the framework of positive free logic is accepted, any rules whatsoever that satisfy its criteria should be acceptable, and, indeed, must be accepted as logical. But we can't accept $\mathbf{FDH}$. 

Suppose we permit the discharge of $\exists !t$ above the third premises of $(\invertediota I)$ to restore harmony with $(\invertediota E_2)$. The resulting theory of $\invertediota$, call it $\mathbf{FDI}$, is axiomatised by the left to right half of \rf{LA} and 

\lbp{FDI}{$FDI$}{$(\forall y(y=t\rightarrow A_y^x)\land (\exists !t\rightarrow \forall y(A_y^x\rightarrow y=t))\rightarrow \invertediota xA=t$}

\noindent Replacing $A$ by $x=t$ proves \rf{C}. If $\neg \exists !t$, the antecedent is trivially fulfilled, so $\vdash \neg \exists !t\rightarrow \invertediota xA=t$, for any $t$ and $A$.  From these by $\invertediota x(x=t_2)=t_2, \invertediota x(x=t_2)=t_1\vdash t_1=t_2$, $\neg \exists !t_1\vdash t_1=t_2$, hence $\neg \exists !t_1\vdash \neg\exists !t_2$, for any $t_1$ and $t_2$. Replacing and contraposing, $\exists !a\vdash \exists !\invertediota x(x\not =x)$. So by the left to right half of \rf{LA} $\invertediota x(x\not=x)=\invertediota x(x\not=x)\rightarrow \forall y(y\not=y\leftrightarrow y=\invertediota x(x\not=x))$. Detaching and instantiating once more with $\invertediota x(x\not =x)$ gives a contradiction. Hence $\vdash \neg\exists !a$.\footnote{This is another result due to a referee.} This is no good either. 

We have not found a viable theory of definite descriptions in positive free logic in which $\invertediota$ is governed by harmonious rules. A number of negative conclusions could be drawn from this: so much the worse for positive free logic; so much the worse for $\invertediota$; so much the worse for formalising definite descriptions as singular terms. In the absence of a thorough philosophical inquiry any such conclusion would be rash. We are faced with a philosophical problem to the investigation of which I invite any interested parties to contribute. 

\medskip

\noindent \textbf{Acknowledgements.}

\noindent A big thank you goes to the referees for this paper: they did an exceptional job in pointing out errors in the initial versions of this paper and providing constructive criticism. This is what refereeing is there for. 

The research in this paper was funded by the European Union (ERC, ExtenDD, project number: 101054714). Views and opinions expressed are however those of the author(s) only and do not necessarily reflect those of the European Union or the European Research Council. Neither the European Union nor the granting authority can be held responsible for them.

\medskip

\bibliographystyle{eptcs}
\bibliography{KurbisTableauxDD}
\end{document}